\documentclass[12pt,a4paper]{article}
\usepackage{graphicx}
\usepackage{amsmath,amsthm,amsfonts} %% pmatrix proof environments and bold math
\usepackage{txfonts,bm}

    \DeclareMathOperator\GF{GF}
    
    \DeclareMathOperator\diag{diag}
    \DeclareMathOperator\PG{PG}
    
    \DeclareMathOperator\Trans{{T}}
    \DeclareMathOperator\Char{Char}

\newcommand{\cP}{{\mathcal P}} %% math calligraphic letters
\newcommand{\cQ}{{\mathcal Q}}

\newcommand{\bP}{{\mathbb P}}

\newcommand{\vG}{{\bm G}}%% math bold letters

\newcommand{\vP}{{\bm P}}

\newcommand{\XX}{\phantom{-}} %% for centering purposes, provides space like a minus sign

\newcommand{\spmatrix}[1]%%
{\mbox{\scriptsize\setlength\arraycolsep{0.5\arraycolsep}$\begin{pmatrix}#1\end{pmatrix}$}}

\newtheorem{lem}{Lemma}
\newtheorem{thm}{Theorem}
{\theoremstyle{definition}

\newtheorem{exa}{Example}
} {\theoremstyle{remark}
\newtheorem{rem}{Remark}
}

\newenvironment{abc}{%
\begin{enumerate}}{\end{enumerate}}

\begin{document}

\title{M\"{o}bius Pairs of Simplices and Commuting Pauli Operators}

%%%%%%%%%%%%%%%%% 1st author %%%%%%%%%%%%%%%%%%%%%
\author{Hans Havlicek$^{*}$ \and Boris Odehnal \and Metod Saniga\thanks{Fellow of the Center for Interdisciplinary Research (ZiF), University of Bielefeld, Germany.}}

\maketitle

\begin{abstract}
There exists a large class of groups of operators acting on Hilbert spaces,
where commutativity of group elements can be expressed in the geometric
language of symplectic polar spaces embedded in the projective spaces PG($n,
p$), $n$ being odd and $p$ a prime. Here, we present a result about commuting
and non-commuting group elements based on the existence of so-called M\"{o}bius
pairs of $n$-simplices, i.~e., pairs of $n$-simplices which are \emph{mutually
inscribed and circumscribed} to each other. For group elements representing an
$n$-simplex there is no element outside the centre which commutes with all of
them. This allows to express the dimension $n$ of the associated polar space in
group theoretic terms. Any M\"{o}bius pair of $n$-simplices according to our
construction corresponds to two disjoint families of group elements (operators)
with the following properties: (i) Any two distinct elements of the same family
do not commute. (ii) Each element of one family commutes with all but one of
the elements from the other family. A three-qubit generalised Pauli group
serves as a non-trivial example to illustrate the theory for $p=2$ and $n=5$.
\end{abstract}

\section{Introduction}

The last two decades have witnessed a surge of interest in the exploration of
the properties of certain groups relevant for physics in terms of finite
geometries. The main outcome of this initiative was a discovery of a large
family of groups -- Dirac and Pauli groups -- where commutativity of two
distinct elements admits a geometrical interpretation in terms of the
corresponding points being joined by an isotropic line (symplectic polar
spaces, see \cite{huppert-67}, \cite{shaw-95}, \cite{saniga+planat-07a},
\cite{planat+saniga-08a}, \cite{sengupta-09a}, \cite{rau-09a}, \cite{thas-09z},
\cite{thas-09y}, and \cite{havlicek+o+s-09z} for a comprehensive list of
references) or the corresponding unimodular vectors lying on the same free
cyclic submodule (projective lines over modular rings, e.~g.,
\cite{havlicek+saniga-07a}, \cite{havlicek+saniga-08a}). This effort resulted
in our recent paper \cite{havlicek+o+s-09z}, where the theory related to polar
spaces was given the most general formal setting.
\par
Finite geometries in general, and polar spaces in particular, are endowed with
a number of remarkable properties which, in light of the above-mentioned
relations, can be directly translated into group theoretical language. In this
paper, our focus will be on one of them. Namely, we shall consider pairs of
$n$-simplices of an $n$-dimensional projective space ($n$ odd) which are
mutually inscribed and circumscribed to each other. First, the existence of
these so-called M\"{o}bius pairs of $n$-simplices will be derived over an arbitrary
ground field. Then, it will be shown which group theoretical features these
objects entail if restricting to finite fields of prime order $p$. Finally, the
case of three-qubit Pauli group is worked out in detail, in view of also
depicting some distinguished features of the case $p=2$.

\section{M\"{o}bius pairs of simplices}\label{se:moebius}

We consider the $n$-dimensional projective space $\PG(n,F)$ over any field $F$,
where $n\geq 1$ is an odd number. Our first aim is to show an $n$-dimensional
analogue of a classical result by M\"{o}bius \cite{moebius-28a}. Following his
terminology we say that two $n$-simplices of $\PG(n,F)$ are \emph{mutually
inscribed and circumscribed\/} if each point of the first simplex is in a
hyperplane of the second simplex, and \emph{vice versa\/} for the points of the
second simplex. Two such $n$-simplices will be called a \emph{M\"{o}bius pair of
simplices} in $\PG(n,F)$ or shortly a \emph{M\"{o}bius pair}. There is a wealth of
newer and older results about M\"{o}bius pairs in $\PG(3,F)$. See, among others,
\cite{guinand-77a}, \cite[p.~258]{coxeter-89}, \cite{witczynski-94a},
\cite{witczynski-97a}. The possibility to find M\"{o}bius pairs of simplices in any
odd dimension $n\geq 3$ is a straightforward task \cite[p.~188]{brau-76-2}:
Given any $n$-simplex in $\PG(n,F)$ take the image of its hyperplanes under any
null polarity $\pi$ as second simplex. By this approach, it remains open,
though, whether or not the simplices have common vertices. For example, if one
hyperplane of the first simplex is mapped under $\pi$ to one of the vertices of
the first simplex, then the two simplices share a common point. However, a
systematic account of the $n$-dimensional case seems to be missing. A few
results can be found in \cite{berzolari-06} and \cite{herrmann-52a}. There is
also the possibility to find M\"{o}bius pairs which are not linked by a null
polarity. See \cite[p.~137]{berzolari-06} for an example over the real numbers
and \cite[p.~290ff.]{cox-58} for an example over the field with three elements.
Other examples arise from the points of the Klein quadric representing a
\emph{double six} of lines in $\PG(3,F)$. See \cite[p.~31]{hirschfeld-85}

\par
We focus our attention to \emph{non-degenerate\/} M\"{o}bius pairs. These are pairs
of $n$-simplices such that each point of either simplex is incident with
\emph{one and only one\/} hyperplane of the other simplex. This property
implies that each point of either simplex does not belong to any subspace which
is spanned by less than $n$ points of the other simplex, for then it would
belong to at least two distinct hyperplanar faces. We present a construction of
non-degenerate M\"{o}bius pairs which works over any field $F$. The problem of
finding \emph{all\/} M\"{o}bius pairs in $\PG(n,F)$ is not within the scope of this
article.
\par
In what follows we shall be concerned with matrices over $F$ which are composed
of the matrices
\begin{equation}\label{eq:matrizen}
    K:=\begin{pmatrix}
     0& -1\\
     1 & \XX 0
    \end{pmatrix},\;\;
    J:=\begin{pmatrix}
     1& 1\\
     1 & 1
    \end{pmatrix},\;\;
    L:=\begin{pmatrix}
     \XX 1& -1\\
       -1 & \XX 1
    \end{pmatrix},
\end{equation}
and the $2\times 2$ unit matrix $I$. We define a null polarity $\pi$ of
$\PG(n,F)$ in terms of the alternating $(n+1)\times (n+1)$ matrix\footnote{Note
that indices range from $0$ to $n$.}
\begin{equation}\label{eq:A}
    A:=\begin{pmatrix}
     K& -J &  \ldots &-J\\
     J & K  & \ldots &-J\\
     \vdots & \vdots   &  \ddots & \vdots\\
     J & J   & \ldots  & K
    \end{pmatrix}.
\end{equation}
Thus all entries of $A$ above the diagonal are $-1$, whereas those below the
diagonal are $1$. Using the identities $-K^2=I$, $JK-KL=0$, and $JL=0$ it is
easily verified that $A$ is indeed an invertible matrix, because
\begin{equation}\label{eq:A-inv}
    A^{-1}=\begin{pmatrix}
    -K& -L &  \ldots    &  -L\\
     L & -K  & \ldots   &  -L\\
     \vdots & \vdots    &  \ddots & \vdots \\
     L & L   & \ldots   &  -K
    \end{pmatrix}.
\end{equation}
\par
Let $\cP:=\{P_0,P_1,\ldots,P_n\}$ be the $n$-simplex which is determined by the
vectors $e_0,e_1,\ldots,e_n$ of the standard basis of $F^{n+1}$, i.~e.,
\begin{equation}\label{eq:P_j}
    P_j= Fe_j\mbox{~~for all~~}j\in\{0,1,\ldots,n\}.
\end{equation}
The elements of $ F^{n+1}$ are understood as column vectors. We first exhibit
the image of $\cP$ under the null polarity $\pi$.

\begin{lem}\label{lem:1}
Let $S$ be a subspace of $\PG(n,F)$ which is generated by $k+1\geq 2$ distinct
points of the simplex $\cP$. Then the following assertions hold:
\begin{abc}
\item
$S\cap \pi(S)=\emptyset$ if $k$ is odd.

\item
$S\cap \pi(S)$ is a single point, which is in general position to the chosen
points of $\cP$, if $k$ is even.

\end{abc}
\end{lem}

\begin{proof}
Suppose that $S$ is the span of the points $P_{j_0}, P_{j_1}, \ldots, P_{j_k}$,
where $0\leq j_0<j_1<\cdots<j_k\leq n$. A point $Y$ is in $S$ if, and only if,
it is represented by a vector $y\in F^{n+1}$ of the form
\begin{equation}\label{}
    y = \sum_{i=0}^k y_{j_i}e_{j_i}\neq 0.
\end{equation}
The rows of $A$ with numbers $j_0, j_1, \ldots, j_k$ comprise the coefficients
of a system of linear equations in $n+1$ unknowns $x_0,x_1,\ldots,x_n$ whose
solutions are the vectors of $\pi(S)$. Substituting the vector $y$ into this
system gives the homogeneous linear system (written in matrix form)
\begin{equation}\label{eq:schnitt}
    A_{k}\cdot (y_{j_0},y_{j_1},\ldots,y_{j_k})^{\Trans} =
    (0,0,\ldots,0)^{\Trans},
\end{equation}
where $A_k$ is the principal submatrix of $A$ which arises from the first $k+1$
rows and columns of $A$. Note that (\ref{eq:schnitt}) holds, because the matrix
$A_k$ coincides with the principal submatrix of $A$ which arises from the rows
and columns with indices $j_0,j_1,\ldots,j_k$. The solutions of
(\ref{eq:schnitt}) are the vectors of $S\cap \pi(S)$. There are two cases:
\par
$k$ odd: Here $A_k$ has full rank $k+1$, as follows by replacing $n$ with $k$
in (\ref{eq:A}) and (\ref{eq:A-inv}). Hence the system (\ref{eq:schnitt}) has
only the zero-solution and $S\cap \pi(S)=\emptyset$, as asserted.
\par
$k$ even: Here $A_k$ cannot be of full rank, as it is an alternating matrix
with an odd number of rows. By the above, the submatrix $A_{k-1}$ has rank $k$,
so that the rank of $A_k$ equals to $k$. This implies that the solutions of the
linear system (\ref{eq:schnitt}) is the span of a single non-zero vector. For
example,
\begin{equation}\label{eq:loesung}
    (-1,1,-1,1,-1,\ldots,-1)^{\Trans} \in F^{k+1}
\end{equation}
is such a vector. It determines the point
%\begin{equation}\label{eq:schnittpunkt}\renewcommand\arraystretch{0.5}
%    \begin{array}{c@{\,}c@{\,}c@{\,}c@{\,}c@{\,}c@{\,}c@{\,}c@{\,}c@{\,}c@{\,}c}
%    P_{{j_0},{j_1},\ldots,{j_k}} := F(\ldots,&1     &,\ldots,&-1    &,\ldots,& 1    &,\ldots,& -1    &,\ldots,& 1    &, \ldots)^{\Trans},\\
%                                             &_{j_0}&        &_{j_1}&        &_{j_2}&        &_{j_3}&        &_{j_k}&
%    \end{array}
%\end{equation}
%where the dots represent zero entries.
\begin{equation}\label{eq:P_neu}
    P_{{j_0},{j_1},\ldots,{j_k}}
    :=
    F\left( \sum_{i=0}^k (-1)^{i+1}e_{j_i}    \right).
\end{equation}
Since the coordinates of $P_{{j_0},{j_1},\ldots,{j_k}}$ with numbers ${j_0},
{j_1}, \ldots, {j_k}$ are non-zero, the points $P_{j_0},P_{j_1},\ldots,P_{j_k},
P_{{j_0},{j_1},\ldots,{j_k}}$ are in general position.
\end{proof}
The previous lemma holds trivially for $k+1=0$ points, since then
$S=\emptyset$. It is also valid, \emph{mutatis mutandis}, in the case $k+1=1$,
even though here one has to take into account $S=P_{j_0}$ yields again the
point $S\cap \pi(S)=P_{j_0}$. Hence the null polarity $\pi$ and the simplex
$\cP$ give rise to the following points: $P_0,P_1\ldots,P_n$ (the points of
$\cP$), $P_{012},P_{013},\ldots,P_{n-2,n-1,n}$ (one point in each plane of
$\cP)$, \ldots, $P_{0,1,\ldots,n-1}, \ldots, P_{1,2,\ldots,n}$ (one point in
each hyperplane of $\cP$). All together these are
\begin{equation}\label{}
    {{n+1}\choose 1} + {{n+1}\choose 3} + \cdots +{{n+1}\choose {n-1}}
    = \sum_{i=0}^{n}{n \choose i} =2^{n}
\end{equation}
mutually distinct points. We introduce another notation by defining
\begin{equation}\label{eq:Q_j}
    P_{{j_0},{j_1},\ldots,{j_k}}=:Q_{{m_0},{m_1},\ldots,{m_{n-k}}},
\end{equation}
where $0\leq {m_0} < {m_1} < \cdots < {m_{n-k}} \leq n$ are those indices which
do not appear in $P_{{j_0},{j_1},\ldots,{j_k}}$. We are now in a position to
state our first main result:

\begin{thm}\label{thm:1}
In $\PG(n,F)$, $n\geq 3$, let the null-polarity $\pi$ and the $n$-simplex
$\cP=\{P_0,P_1,\ldots,P_n\}$ be given according to \emph{(\ref{eq:A})} and
\emph{(\ref{eq:P_j})}, respectively. Then the following assertions hold:

\begin{abc}
\item\label{thm:1a}
$\cP$ and $\cQ:=\{Q_0,Q_1,\ldots,Q_n\}$, where the points $Q_m$ are defined by
\emph{(\ref{eq:Q_j})}, is a non-degenerate M\"{o}bius pair of $n$-simplices.

\item\label{thm:1b}
The $n$-simplices $\cP$ and $\cQ$ are in perspective from a point if, and only
if, $F$ is a field of characteristic two.
\end{abc}
\end{thm}
\begin{proof}
Ad~(a): Choose any index $m\in\{0,1,\ldots,n\}$. Then $Q_m$ is the image under
$\pi$ of the hyperplane $S$ which is spanned by
$P_0,\ldots,P_{m-1},P_{m+1},\ldots,P_n$. The proof of Lemma~\ref{lem:1} shows
how to find a system of linear equations for $Q_m$. Furthermore, formula
(\ref{eq:P_neu}) provides a coordinate vector for $Q_m$. However, such a vector
can be found directly by extracting the $m$-th column of the matrix $A^{-1}$,
viz.\
\begin{equation}\label{eq:q_m}
    \sum_{i=0}^{m-1}(-1)^{i+m+1} e_i + \sum_{k=m+1}^{n}(-1)^{k+m} e_k =: q_m.
\end{equation}
(The vector from (\ref{eq:P_neu}) is $(-1)^m q_m$.) As the columns of $A^{-1}$
form a basis of $F^{n+1}$, the point set $\cQ$ is an $n$-simplex.
\par
The $n+1$ hyperplanes of the simplex $\cQ$ are the images under $\pi$ of the
vertices of $\cP$. Each of these hyperplanes has a linear equation whose
coefficients comprise one of the rows of the matrix $A$. So a point $P_j$ is
incident with the hyperplane $\pi(P_i)$ if, and only if, the $(i,j)$-entry of
$A$ is zero. Since each row of $A$ has precisely one zero entry, we obtain that
each point of $\cP$ is incident with one and only one hyperplane of $\cQ$.
\par
In order to show that each point of $\cQ$ is incident with precisely one
hyperplane of the simplex $\cP$, we apply a change of coordinates from the
standard basis $e_0,e_1,\ldots,e_n$ to the basis $b_i:=(-1)^{i}q_i$,
$i\in\{0,1,\ldots,n\}$. The points $Fb_i$ constitute the $n$-simplex $\cQ$. Let
$B$ be the matrix with columns $b_0,b_1,\ldots,b_n$. With respect to the basis
$b_i$ the columns of $B^{-1}$ describe the points of $\cP$, and
$B^{\Trans}AB=A$ is a matrix for $\pi$. The columns of $B^{-1}$ and $A^{-1}$
are identical up to an irrelevant change of signs in columns with odd indices.
Therefore, with respect the basis $b_i$, the simplex $\cQ$ plays the role of
$\cP$, and \emph{vice versa}. So the assertion follows from the result in the
preceding paragraph.
\par
Ad~(b): For each $j\in\{0,2,\ldots,n-1\}$ the lines $P_jQ_j$ and
$P_{j+1}Q_{j+1}$ meet at that point which is given by the vector
\begin{equation}\label{eq:punkt}
    -e_j + q_j = -(e_{j+1} + q_{j+1})
    =(-1,1,\ldots,-1,1,\underbrace{-1,-1}_{j,\,j+1},1,-1,\ldots,1,-1)^{\Trans} .
\end{equation}
Comparing signs we see that $-e_0+q_0$ and $-e_2+q_2$ are linearly independent
for $\Char K \neq 2$, whereas for $\Char K = 2$ all lines $P_kQ_k$,
$k\in\{0,1,\ldots,n\}$ concur at the point
\begin{equation}\label{eq:C}
    C:= F(1,1,\ldots,1)^{\Trans} .
\end{equation}
\end{proof}

\begin{rem}\label{rem:1}
Choose $k+1$ distinct vertices of $\cP$, where $3\leq k\leq n$ is odd. Up to a
change of indices it is enough to consider the $k$-simplex
$\{P_0,P_1,\ldots,P_k\}$ and its span, say $S$. The null polarity $\pi$ induces
a null polarity $\pi_S$ in $S$ which assigns to $X\in S$ the
$(k-1)$-dimensional subspace $\pi(X)\cap S$. We get within $S$ the settings of
Theorem~\ref{thm:1} with $k$ rather than $n$ points and $\pi_S$ instead of
$\pi$. The \emph{nested\/} non-degenerate M\"{o}bius pair in $S$ is formed by the
$k$-simplex $\{P_0,P_1,\ldots,P_k\}$ and the $k$-simplex comprising the points
\begin{equation}\label{eq:nest}
     Q_{0,k+1,\ldots,n},\; %% = P_{1,2,3,\ldots,k}
     Q_{1,k+1,\ldots,n},\; %% = P_{0,2,3,\ldots,k}
        \ldots,\;
     Q_{k,k+1,\ldots,n}. %% = P_{0,1,2,\ldots,k-2,}
\end{equation}
This observation illustrates the meaning of all the $2^n$ points which arise
from $\cP$ and the null polarity $\pi$. If we allow $k=1$ in the previous
discussion then, to within a change of indices, the \emph{nested\/} degenerate
M\"{o}bius pair $\{P_0,P_1\}=\{Q_{0,2,\ldots,n}, Q_{1,2,\ldots,n}\}$ is obtained.
\end{rem}

\begin{rem}\label{rem:2}
The case $F=\GF(2)$ deserves particular mention, since we can give an
interpretation for \emph{all points\/} of $\PG(n,2)$ in terms of our M\"{o}bius
pair. Recall the following notion from the theory of binary codes: The
\emph{weight\/} of an element of $\GF(2)^{n+1}$ is the number of $1$s amongst
its coordinates. The $2^n$ points addressed in Remark~\ref{rem:1} are given by
the vectors with \emph{odd weight}. The $2^n$ vectors with \emph{even weight}
are, apart from the zero vector, precisely those vectors which yield the
$2^{n}-1$ points of the hyperplane $\pi(C): \sum_{i=0}^{n+1}x_i=0$. More
precisely, the vectors with even weight $w\geq 4$ are the centres of
perspectivity for the nested non-degenerate M\"{o}bius pairs of $(w-1)$-simplices,
whereas the vectors with weight $2$ are the points of intersection of the edges
of $\cP$ with $\pi(C)$. The latter points may be regarded as ``centres of
perspectivity'' for the degenerate M\"{o}bius pairs formed by the two vertices of
$\cP$ on such an edge. Each point of the hyperplane $\pi(C)$ is the centre of
perspectivity of precisely one nested M\"{o}bius pair.

\end{rem}

\section{Commuting and non-commuting elements}\label{se:komm}

Our aim is to translate the properties of M\"{o}bius pairs into properties of
commuting and non-commuting group elements. We shortly recall some results from
\cite{havlicek+o+s-09z}. Let $(\vG,\cdot)$ be a group and $p$ be a prime.
Suppose that the centre $Z(\vG)$ of $\vG$ contains the commutator subgroup
$\vG'=[\vG,\vG]$ and the set $\vG^{(p)}$ of $p$th powers. Also, let $\vG'$ be
of order $p$. Then $V:=\vG/Z(\vG)$ is a commutative group which, if written
additively, is a vector space over $\GF(p)$ in a natural way. Furthermore,
given any generator $g$ of $\vG'$ we have a bijection $\psi_g : \vG'\to\GF(p):
g^m \mapsto m$ for all $m\in\{0,1,\ldots,p\}$. The commutator function in $G$
assigns to each pair $(x,y)\in \vG\times \vG$ the element
$[x,y]=xyx^{-1}y^{-1}$. It gives rise to the non-degenerate alternating
bilinear form
\begin{equation}\label{eq:blf}
    [\cdot,\cdot]_g : V\times V\to \GF(p) : (xZ(\vG),yZ(\vG)) \mapsto
    \psi_g([x,y]).
\end{equation}
We assume now that $V$ has finite dimension $n+1$, and we consider the
projective space $\PG(n,p):=\bP(V)$. Its points are the one-dimensional
subspaces of $V$. In our group theoretic setting a non-zero vector of $V$ is a
coset $xZ(\vG)$ with $x\in\vG\setminus Z(\vG)$. The scalar multiples of
$xZ(\vG)$ are the cosets of the form $x^k Z(\vG)$, $k\in\{0,1,\ldots,p-1\}$,
because multiplying a vector of $V$ by $k\in\GF(p)$ means taking a $k$th power
in $\vG/Z(\vG)$. So $x,x'\in\vG$ describe the same point $X$ of $\PG(n,p)$ if,
and only if, none of them is in the centre of $\vG$, and $x'=x^k z$ for some
$k\in\{1,2,\ldots,p-1\}$ and some $z\in Z(\vG)$. Under these circumstances $x$
(and likewise $x'$) is said to \emph{represent\/} the point $X$. Conversely,
the point $X$ is said to \emph{correspond\/} to $x$ (and likewise $x'$). Note
that the elements of $Z(\vG)$ determine the zero vector of $V$. So they do not
represent any point of $\PG(n,p)$.
\par
The non-degenerate alternating bilinear form from (\ref{eq:blf}) determines a
null polarity $\pi$ in $\PG(n,p)$. We quote the following result from
\cite[Theorem~6]{havlicek+o+s-09z}: \emph{Two elements $x,y\in \vG\setminus
Z(\vG)$ commute if, and only if, their corresponding points in $\PG(n,p)$ are
conjugate with respect to $\pi$, i.~e., one of the points is in the polar
hyperplane of the other point.} This crucial property is the key for proving
Lemma~\ref{lem:2} and Theorem~\ref{thm:2} below.

\begin{lem}\label{lem:2}
Suppose that $x_0,x_1,\ldots,x_r\in\vG\setminus Z(\vG)$ is a family of group
elements. Then the following assertions are equivalent:
\begin{abc}
\item\label{lem:2a}
The points corresponding to $x_0,x_1,\ldots,x_r$ constitute an $n$-simplex of
the projective space $\PG(n,p)$, whence $r=n$.

\item\label{lem:2b}
There exists no element in $\vG\setminus Z(\vG)$ which commutes with all of
$x_0,x_1,\ldots,x_r$, but for each proper subfamily of $x_0,x_1,\ldots,x_r$ at
least one such element exists.
\end{abc}
\end{lem}

\begin{proof}
The points corresponding to the family $(x_i)$ generate $\PG(n,p)$ if, and only
if, their polar hyperplanes have no point in common. This in turn is equivalent
to the non-existence of an element in $\vG\setminus Z(\vG)$ which commutes with
all elements of the family $(x_i)$. The proof is now immediate from the
following observation: An $n$-simplex of $\PG(n,p)$ can be characterised as
being a minimal generating family of $\PG(n,p)$.
\end{proof}
This result shows that the dimension $n+1$ of $V$ can be easily determined by
counting the cardinality of a family of group elements which satisfies
condition (\ref{lem:2b}). We close this section by translating
Theorem~\ref{thm:1}:

\begin{thm}\label{thm:2}
Let $\vG$ be a group which satisfies the assumptions stated in the first
paragraph of this section. Also, let $V=\vG/Z(\vG)$ be an $(n+1)$-dimensional
vector space over $\GF(p)$. Suppose that $x_0,x_1,\ldots,x_n\in\vG\setminus
Z(\vG)$ and $y_0,y_1,\ldots,y_n\in\vG\setminus Z(\vG)$ are two families of
group elements which represent a non-degenerate M\"{o}bius pair of\/ $\PG(n,p)$ as
in Theorem~\emph{\ref{thm:1}}. Then the following assertions hold:
\begin{abc}

\item\label{thm:2a}
There exists no element in $\vG\setminus Z(\vG)$ which commutes with
$x_0,x_1,\ldots,x_n$.

\item\label{thm:2b}
The elements $x_0,x_1,\ldots,x_n$ are mutually non-commuting.

\item\label{thm:2c}
For each $i\in\{1,2,\ldots,n\}$ the element $x_i$ commutes with all $y_j$ such
that $j\neq i$.
\end{abc}
Each of these three assertions remains true when changing the role of the
elements $x_0,x_1,\ldots,x_n$ and $y_0,y_1,\ldots,y_n$.
\end{thm}
\begin{proof}
The assertion in (\ref{thm:2a}) follows from Lemma~\ref{lem:2}. Since all
non-diagonal entries of the matrix $A$ from (\ref{eq:A}) equal to $1$, no two
points which are represented by the elements $x_i$ are conjugate with respect
to $\pi$. Hence (\ref{thm:2b}) is satisfied. Finally, (\ref{thm:2c}) follows,
as the polar hyperplane of the point represented by $x_i$ contains all the
points which are represented by the elements $y_j$ with $j\neq i$. The last
statement holds due to the symmetric role of the two simplices of a M\"{o}bius pair
which was established in the proof of Theorem~\ref{thm:1}~(\ref{thm:1a}).
\end{proof}

\begin{rem}\label{rem:3}
According to Remark~\ref{rem:1} we may obtain nested non-degenerate M\"{o}bius
pairs from appropriate subfamilies of $x_0,x_1,\ldots,x_n$. These M\"{o}bius pairs
satisfy, \emph{mutatis mutandis}, properties (\ref{thm:2b}) and (\ref{thm:2c}).
\end{rem}

\begin{exa}\label{exa:1}
We consider the complex matrices
\begin{equation}\label{eq:pauli}
    \sigma_0:=\begin{pmatrix}1 &
    \hphantom{-}0\\0&\hphantom{-}1\end{pmatrix},\;\;
    \sigma_x:=\begin{pmatrix}0 &
    \hphantom{-}1\\1&\hphantom{-}0\end{pmatrix},\;\;
    \sigma_y:=\begin{pmatrix}0 & -i\\i&\hphantom{-}0\end{pmatrix},\;\;
    \sigma_z:=\begin{pmatrix}1 & \hphantom{-}0 \\0&-1\end{pmatrix}.
\end{equation}
The matrices $i^\alpha\sigma_\beta$ with $\alpha\in\{0,1,2,3\}$ and $\beta \in
\{0,x,y,z\}$ constitute the \emph{Pauli group\/} $\vP$ of order $16$. The
centre of $\vP$ is $Z(\vP)=\big\{ i^\alpha\sigma_0\mid \alpha\in\{0,1,2,3\}
\big\}$. The commutator subgroup $\vP'=\{\pm\sigma_0\}$ and the set
$\vP^{(2)}=\{\pm\sigma_0\}$ of squares are contained in $Z(\vP)$. By
Section~\ref{se:komm}, the factor group $\vP/Z(\vP)$, if written additively, is
a vector space over $\GF(2)$. For each $\beta \in \{0,x,y,z\}$ the coset
$Z(\vP)\sigma_\beta$ is denoted by $\beta$. In this notation, addition can be
carried out according to the relations $0+\beta=\beta$, $\beta+\beta=0$, and
$x+y+z=0$. The mapping
\begin{equation}\label{eq:koord}
    0\mapsto (0,0)^{\Trans},\;\; x\mapsto (1,0)^{\Trans},\;\;
    y\mapsto (0,1)^{\Trans},\;\; z\mapsto (1,1)^{\Trans}
\end{equation}
is an isomorphism of $\vP/Z(\vP)$ onto the additive group of the vector space
$\GF(2)\times \GF(2)$.
\par
Let $\vG$ be the group of order $256$ comprising the three-fold Kronecker
products $i^\alpha\sigma_\beta\otimes\sigma_\gamma\otimes\sigma_\delta$ with
$\alpha\in\{0,1,2,3\}$ and $\beta,\gamma,\delta\in\{0,x,y,z\}$. This group acts
on the eight-dimensional Hilbert space of three qubits. In our terminology from
Section~\ref{se:komm} (with $p:=2$) we have
\begin{equation}\label{}
    Z(\vG) = \big\{ i^\alpha \sigma_0\otimes\sigma_0\otimes\sigma_0 \mid  \alpha\in\{0,1,2,3\}  \big\},
    \;\;\;\;
    \vG'=\vG^{(2)}=\{\pm\sigma_0\otimes\sigma_0\otimes\sigma_0\},
\end{equation}
and $g=-\sigma_0\otimes\sigma_0\otimes\sigma_0$. Hence $V=\vG/Z(\vG)$ is a
six-dimensional vector space over $\GF(2)$ endowed with an alternating bilinear
form $[\cdot,\cdot]_g$. We introduce $\beta\gamma\delta$ as a shorthand for
$Z(\vG)(\sigma_\beta\otimes\sigma_\gamma\otimes\sigma_\delta)$, where
$\beta,\gamma,\delta\in\{0,x,y,z\}$. In this notation, addition in $V$ can be
carried out componentwise according to the relations stated before. An
isomorphism of $V$ onto $\GF(2)^6$ is obtained by replacing the three symbols
of an element of $V$ according to (\ref{eq:koord}). This gives the
\emph{coordinate vector\/} of an element of $V$. For example, the coordinate
vectors of the six elements
\begin{equation}\label{eq:basis}
    x00, y00,0x0, 0y0, 00x,00y\in V
\end{equation}
comprise the standard basis of $\GF(2)^6$. These six elements therefore form a
basis of $V$. The projective space $\PG(5,2)=\bP(V)$, like any projective space
over $\GF(2)$, has the particular property that each of its points is
represented by one and only one non-zero vector of $V$. \emph{We therefore
identify $V\setminus\{000\}$ with $\PG(5,2)$.}
\par
Recall the matrices defined in (\ref{eq:matrizen}). The matrix of the
alternating bilinear form from (\ref{eq:blf}) with respect to the basis
(\ref{eq:basis}) equals to the $6 \times 6$ matrix $ \diag(K,K,K)$ over
$\GF(2)$. In order to obtain a M\"{o}bius pair
\begin{equation}\label{}
    \cP=\{P_0,P_1,\ldots,P_5\},\;\;\;\; \cQ=\{Q_0,Q_1,\ldots,Q_5\}
\end{equation}
we have to use another basis of $V$, e.~g., the one which arises in terms of
coordinates from the six columns of the matrix
\begin{equation}\label{eq:trafo}
    T:=\begin{pmatrix}
    I & J & J\\
    0 & I & J\\
    0 & 0 & I
    \end{pmatrix}.
\end{equation}
Indeed, $T^{\Trans}\cdot\diag(K,K,K)\cdot T$ gives an alternating $6\times 6$
matrix $A$ as in (\ref{eq:A}). We thus can translate our results from
Section~\ref{se:moebius} as follows: First, we multiply $T$ with the ``old''
coordinate vectors from there and, second, we express these ``new'' coordinate
vectors as triplets in terms of $0,x,y,z$. The vertices $P_0,P_1,\ldots,P_5$
and $Q_0,Q_1,\ldots,Q_5$ can be read off, respectively, from the first and
second row of the following matrix:
\begin{equation}\label{eq:PQ5}
    \begin{array}{c@{\;\;}c@{\;\;}c@{\;\;}c@{\;\;}c@{\;\;}c}
        x00&y00&zx0&zy0&zzx&zzy \\
        yzz&xzz&0yz&0xz&00y&00x
    \end{array}
\end{equation}
We note that $\cP$ and $\cQ$ are in perspective from a point according to
Theorem~\ref{thm:1}. This point is $zzz$. Since each line of $\PG(5,2)$ has
only three points, the entries of the second row in (\ref{eq:PQ5}) can be found
by adding $zzz$ to the entries from the first row. Each point of $\PG(5,2)$
corresponds to four elements of the group $\vG$, whence the points of
$\cP\cup\cQ$ correspond to $48$ elements of $\vG$, none of them in the centre
$Z(\vG)$. We can rephrase the M\"{o}bius property as follows: \emph{Let two out of
these $48$ elements of $\vG$ represent distinct points. Then these two elements
commute if, and only if, they represent points which are in distinct rows and
distinct columns of the matrix} (\ref{eq:PQ5}). The $20$ points
$P_{012},P_{013},\ldots,P_{345}$ are obtained by adding three of the points of
$\cP$. Explicitly, we get:
\begin{equation}\label{eq:20}
    \begin{array}{c@{\;\;}c@{\;\;}c@{\;\;}c@{\;\;}c@{\;\;}c@{\;\;}c@{\;\;}c@{\;\;}c@{\;\;}c}
        0x0&0y0&0zx&0zy&xz0&xyx&xyy&xxx&xxy&x0z \\
        yz0&yyx&yyy&yxx&yxy&y0z&z0x&z0y&zxz&zyz
    \end{array}
\end{equation}
We leave it to the reader to find the ${6\choose 4}=15$ nested M\"{o}bius pairs of
tetrahedra which are formed by four points from $\cP$ and the four appropriate
points from (\ref{eq:20}). By Remark~\ref{rem:2}, the $32$ points from
(\ref{eq:PQ5}) and (\ref{eq:20}) are precisely the points off the polar
hyperplane of $zzz$. This means that none of the corresponding elements of
$\vG$ commutes with the representatives of the distinguished point $zzz$.
\end{exa}
The results from \cite{havlicek+o+s-09z} show that Theorem~\ref{thm:2} can be
applied to a wide class of groups, including the generalised Pauli groups
acting on the space of $N$-qudits provided that $d$ is a prime number.

\section{Conclusion}

Following the strategy set up in our recent paper \cite{havlicek+o+s-09z}, we
have got a deeper insight into the geometrical nature of a large class of
finite groups, including many associated with finite Hilbert spaces. This was
made possible by employing the notion of a M\"{o}bius pair of $n$-simplices in a
finite odd-dimensional projective space, PG$(n,p)$, $p$ being a prime.
Restricting to non-degenerate M\"{o}bius pairs linked by a null polarity, we have
first shown their existence for any odd $n$, a remarkable nested structure they
form, and perspectivity from a point of the simplices in any such pair if
$p=2$. Then, the commutation properties of the group elements associated with a
M\"{o}bius pair have been derived. In particular, the two disjoint families of
$n+1$ group elements that correspond to a M\"{o}bius pair are such that any two
distinct elements/operators from the same family do not commute and each
element from one family commutes with all but one of the elements from the
other family. As the theory also encompasses a number of finite generalised
Pauli groups, that associated with three-qubits ($n=5$ and $p=2$) was taken as
an illustrative example, also because of envisaged relevance of M\"{o}bius pairs to
entanglement properties of a system of three fermions with six single-particle
states \cite{levay+vrana-09z}. It should, however, be stressed that
above-outlined theory is based on a particular construction of M\"{o}bius pairs,
and so there remains an interesting challenge to see in which way it can be
generalised to incorporate arbitrary M\"{o}bius pairs.

\subsubsection*{Acknowledgements}

This work was carried out in part within the ``Slovak-Austrian Science and
Technology Cooperation Agreement'' under grants SK 07-2009 (Austrian side) and
SK-AT-0001-08 (Slovak side), being also partially supported by the VEGA grant
agency projects Nos. 2/0092/09 and 2/7012/27. The final version was completed
within the framework of the Cooperation Group ``Finite Projective Ring
Geometries: An Intriguing Emerging Link Between Quantum Information Theory,
Black-Hole Physics, and Chemistry of Coupling'' at the Center for
Interdisciplinary Research (ZiF), University of Bielefeld, Germany.

\footnotesize
%%%%%%%%%%%%%%%%%%%%%% BIBTEX ACTIVE %%%%%%%%%%
%\bibliographystyle{plain} %% {unsrt}
%\bibliography{../../separata/quantum,../../separata/ketten}

\begin{thebibliography}{10}\itemsep-1pt

\bibitem{berzolari-06} L.~Berzolari.
\newblock Sull' estensione del concetto di tetraedri di {M}{\"o}bius agli
  iperspaz\^{\i}.
\newblock {\em Rend. Circ. Mat. Palermo}, 22:136--140, 1906.

\bibitem{brau-76-2} H.~Brauner.
\newblock {\em Geometrie projektiver {R}\"aume II}.
\newblock BI Wissen\-schafts\-ver\-lag, Mannheim, 1976.

\bibitem{cox-58} H.~S.~M. Coxeter.
\newblock Twelve points in {${\rm PG}(5,3)$} with {$95040$}\
  self-transformations.
\newblock {\em Proc.\ Roy.\ Soc.\ London, Ser.\ A}, 247:279--293, 1958.

\bibitem{coxeter-89} H.~S.~M. Coxeter.
\newblock {\em Introduction to Geometry}.
\newblock Wiley Classics Library. John Wiley \& Sons Inc., New York, 1989.
\newblock Reprint of the 1969 edition.

\bibitem{guinand-77a} A.~P. Guinand.
\newblock Graves triads, {M}{\"o}bius pairs, and related matrices.
\newblock {\em J. Geom.}, 10(1-2):9--16, 1977.

\bibitem{havlicek+o+s-09z} H.~Havlicek, B.~Odehnal, and M.~Saniga.
\newblock Factor-group-generated polar spaces and
  {(multi-)}\hspace{0pt}{qudits}.
\newblock {\em SIGMA Symmetry Integrability Geom. Methods Appl.}, submitted.
\newblock (arXiv:0903.5418).

\bibitem{havlicek+saniga-07a} H.~Havlicek and M.~Saniga.
\newblock Projective ring line of a specific qudit.
\newblock {\em J. Phys. A}, 40(43):F943--F952, 2007.
\newblock (arXiv:0708.4333).

\bibitem{havlicek+saniga-08a} H.~Havlicek and M.~Saniga.
\newblock Projective ring line of an arbitrary single qudit.
\newblock {\em J. Phys. A}, 41(1):015302, 12 pp., 2008.
\newblock (arXiv:0710.0941).

\bibitem{herrmann-52a} H.~Herrmann.
\newblock Matrizen als projektive {F}iguren.
\newblock {\em Jber. Deutsch. Math. Verein.}, 56(Abt. 1):6--20, 1952.
\newblock Erratum: ibid. 104.

\bibitem{hirschfeld-85} J.~W.~P. Hirschfeld.
\newblock {\em Finite Projective Spaces of Three Dimensions}.
\newblock Oxford University Press, Oxford, 1985.

\bibitem{huppert-67} B.~Huppert.
\newblock {\em Endliche {G}ruppen. {I}}.
\newblock Die Grundlehren der Mathematischen Wissenschaften, Band 134.
  Springer-Verlag, Berlin, 1967.

\bibitem{levay+vrana-09z} P.~L{\'e}vay and P.~Vrana.
\newblock Three fermions with six single-particle states can be entangled in
  two inequivalent ways.
\newblock {\em Phys. Rev. A}, 78:022329, 10 pp., 2008.
\newblock (arXiv:0903.0541).

\bibitem{moebius-28a} F.~A. M\"{o}bius.
\newblock Kann von zwei dreiseitigen {P}yramiden eine jede in {B}ezug auf die
  andere um- und eingeschrieben zugleich heissen?
\newblock {\em J. reine angew. Math.}, 3:273--278, 1828.

\bibitem{planat+saniga-08a} M.~Planat and M.~Saniga.
\newblock On the {P}auli graphs of {$N$}-qudits.
\newblock {\em Quantum Inf. Comput.}, 8(1-2):127--146, 2008.
\newblock (arXiv:quant-ph/0701211).

\bibitem{rau-09a} A.~R.~P. Rau.
\newblock Mapping two-qubit operators onto projective geometries.
\newblock {\em Phys. Rev. A}, 79:042323, 6 pp., 2009.
\newblock (arXiv:0808.0598).

\bibitem{saniga+planat-07a} M.~Saniga and M.~Planat.
\newblock Multiple qubits as symplectic polar spaces of order two.
\newblock {\em Adv. Stud. Theor. Phys.}, 1(1-4):1--4, 2007.
\newblock (arXiv:quant-ph/0612179).

\bibitem{sengupta-09a} A.~N. Sengupta.
\newblock Finite geometries with qubit operators.
\newblock {\em Infin. Dimens. Anal. Quantum Probab. Relat. Top.},
  12(2):359--366, 2009.
\newblock Preprint (arXiv:0904.2812).

\bibitem{shaw-95} R.~Shaw.
\newblock Finite geometry, {D}irac groups and the table of real {C}lifford
  algebras.
\newblock In {\em Clifford algebras and spinor structures}, volume 321 of {\em
  Math. Appl.}, pages 59--99. Kluwer Acad. Publ., Dordrecht, 1995.

\bibitem{thas-09z} K.~Thas.
\newblock Pauli operators of {$N$}-qubit {Hilbert} spaces and the
  {Saniga}-{Planat} conjecture.
\newblock {\em Chaos, Solitons, Fractals}.
\newblock in press.

\bibitem{thas-09y} K.~Thas.
\newblock The geometry of generalized {Pauli} operators of {$N$}-qudit
  {Hilbert} space, and an application to {MUBs}.
\newblock {\em Europhys. Lett. EPL}, 86:60005, 3 pp., 2009.

\bibitem{witczynski-94a} K.~Witczy{\'n}ski.
\newblock Some remarks on the theorem of {M}\"obius.
\newblock {\em J. Geom.}, 51(1-2):187--189, 1994.

\bibitem{witczynski-97a} K.~Witczy{\'n}ski.
\newblock M\"obius' theorem and commutativity.
\newblock {\em J. Geom.}, 59(1-2):182--183, 1997.

\end{thebibliography}

%%%%%%%%%%%%%%%%%%%%%% BIBTEX OUTPUT %%%%%%%%%%

%%%%%%%%%%%%%%%%%%%%%%%BIBTEX END   %%%%%%%%%%%

\normalsize
%%%%%%%%%%%%%%%%% authors %%%%%%%%%%%%%%%%%%%%%
\noindent
Hans Havlicek and Boris Odehnal\\
Institut f\"{u}r Diskrete Mathematik und Geometrie\\
Technische Universit\"{a}t\\
Wiedner Hauptstra{\ss}e 8--10/104\\
A-1040 Wien, Austria\\
\texttt{havlicek, boris@geometrie.tuwien.ac.at}\\
\\
\noindent Metod Saniga\\
Astronomical Institute\\
Slovak Academy of Sciences\\
SK-05960 Tatransk\'{a} Lomnica\\
Slovak Republic\\
\texttt{msaniga@astro.sk}
%%%%%%%%%%%%%%%%%%%%%%%%%%%%%%%%%%%%%%%%%%%%%%%%%%

\end{document}